\pdfoutput=1
\documentclass[a4paper,UKenglish,cleveref, autoref, thm-restate]{lipics-v2021}

\usepackage{tikz}
\usepackage{verbatim,mathtools}
\usetikzlibrary{positioning,decorations.pathreplacing}

\newcommand{\regL}{\mathsf{L}}

\newcommand{\Int}{\mathbb{Z}}

\newcommand{\ket}[1]{|#1\rangle}
\newcommand{\norm}[1]{\|#1\|}
\newcommand{\normmax}[1]{\|#1\|_\infty}

\newcommand{\E}{\mathbb{E}}

\newcommand{\ceil}[1]{\left\lceil #1 \right\rceil}
\newcommand{\floor}[1]{\left\lfloor #1 \right\rfloor}
\providecommand{\Aa}{\mathcal{A}}
\providecommand{\Cc}{\mathcal{C}}

\providecommand{\Cc}{\mathsf{C}}
\providecommand{\Hh}{\mathcal{H}}
\providecommand{\Hhh}{\mathscr{H}}
\providecommand{\Kk}{\mathcal{K}}
\providecommand{\basis}{\mathcal{B}}

\providecommand{\Bb}{\mathcal{B}}
\providecommand{\Ss}{\mathcal{S}}

\newcommand{\poly}{\mathrm{poly}}

\newcommand{\supp}{\mathrm{supp}}

\providecommand{\Cc}{\mathcal{C}}

\newcommand{\gentrap}{\mathsf{GENTRAP}}
\newcommand{\invert}{\mathsf{INVERT}}
\newcommand{\LWE}[1]{\mathsf{LWE}_{#1}}
\newcommand{\ADD}{\mathsf{ADD}}
\newcommand{\MULT}{\mathsf{MULT}}

\newcommand{\state}{\mathsf{State Generation}}

\title{Test of Quantumness with Small-Depth Quantum Circuits} 

\titlerunning{Test of Quantumness with Small-Depth Quantum Circuits} 

\author{Shuichi Hirahara}{National Institute of Informatics, Japan}{}{}{}

\author{Fran\c{c}ois Le Gall}{Graduate School of Mathematics, Nagoya University, Japan}{}{}{}

\authorrunning{S. Hirahara and F. Le Gall} 

\Copyright{Shuichi Hirahara and Fran\c{c}ois Le Gall} 

\ccsdesc[500]{Theory of computation~Quantum complexity theory}

\keywords{Quantum computing, small-depth circuits, quantum cryptography} 

\funding{JSPS KAKENHI grants Nos.~JP19H04066, JP20H05966, JP20H00579, JP20H04139, JP21H04879 and MEXT Quantum Leap Flagship Program (MEXT Q-LEAP) grants No.~JPMXS0118067394 and JPMXS0120319794. }

\acknowledgements{The authors are grateful to Ryo Hiromasa, Tomoyuki Morimae, Yasuhiko Takahashi and Seiichiro Tani for helpful discussions.}

\nolinenumbers 

\hideLIPIcs  

\EventEditors{Filippo Bonchi and Simon J. Puglisi}
\EventNoEds{2}
\EventLongTitle{46th International Symposium on Mathematical Foundations of Computer Science (MFCS 2021)}
\EventShortTitle{MFCS 2021}
\EventAcronym{MFCS}
\EventYear{2021}
\EventDate{August 23--27, 2021}
\EventLocation{Tallinn, Estonia}
\EventLogo{}
\SeriesVolume{202}
\ArticleNo{19}

\begin{document}

\maketitle

\begin{abstract}
Recently Brakerski, Christiano, Mahadev, Vazirani and Vidick (FOCS 2018) have shown how to construct a test of quantumness based on the learning with errors ($\LWE{}$) assumption: a test that can be solved efficiently by a quantum computer but cannot be solved by a classical polynomial-time computer under the $\LWE{}$ assumption. This test has lead to several cryptographic applications. In particular, it has been applied to producing certifiable randomness from a single untrusted quantum device, self-testing a single quantum device and device-independent quantum key distribution. 

In this paper, we show that this test of quantumness, and essentially all the above applications, can actually be implemented by a very weak class of quantum circuits: constant-depth quantum circuits combined with logarithmic-depth classical computation. This reveals novel complexity-theoretic properties of this fundamental test of quantumness and gives new concrete evidence of the superiority of small-depth quantum circuits over classical computation.
\end{abstract}

\section{Introduction}
\subparagraph{Background.}
A very active research area in quantum computing is proving the superiority of ``weak'' models of quantum computation, such as small-depth quantum circuits, over classical computation. The main motivation is that such models are expected to be much easier to implement than universal quantum computation (e.g., polynomial-size quantum circuits) --- Indeed in the past years we have been witnessing the development of several small-scale quantum computers (see, e.g., \cite{Wikipedia} for information about current quantum computers). 

Under assumptions such as the non-collapse of the polynomial hierarchy or the hardness of (appropriate versions of) the permanent, strong evidence of the superiority of weak classes of quantum circuits has been obtained from the 2000s \cite{Aaronson+STOC11,Aaronson+A14,Aaronson+CCC17,Bouland+18,Bremner+10,Bremner+PRL16,Bremner+17,Fahri+16,Fujii+PRL18,Fujii+16,Morimae+PRL14,Terhal+04}.
A recent breakthrough by Bravyi, Gosset and K\"onig \cite{Bravyi+18}, further strengthened by subsequent works \cite{Bene+18, Bravyi+FOCS19, Coudron+18,LeGallCCC19}, showed an \emph{unconditional} separation between the computational powers of quantum and classical small-depth circuits by exhibiting a computational task that can be solved by constant-depth quantum circuits but requires logarithmic depth for classical circuits. A major shortcoming, however, is that logarithmic-depth classical computation is a relatively weak complexity class. Due to the notorious difficulty of proving superlogarithmic lower bounds on the depth of classical circuits, showing significantly stronger unconditional separations seems completely out of reach of current techniques.

Progress has nevertheless been achieved recently by modifying the concept of computational problem, and considering \emph{interactive problems} (problems consisting of several rounds of interaction between the computational device and a verifier). Grier and Schaeffer \cite{Grier+STOC20}, in particular, showed that there exists an interactive problem that can be solved by constant-depth quantum circuits but such that any classical device solving it would solve $\oplus \regL$-problems. This is a stronger evidence of the superiority of constant-depth quantum circuits since the complexity class $\oplus \regL$ is expected to be significantly larger than logarithmic-depth classical computation. On the other hand, problems in $\oplus \regL$ are still tractable classically since they can be solved in polynomial time.\footnote{More precisely, we have the inclusions $\mathsf{NC}_1\subseteq\regL\subseteq \oplus \regL\subseteq \mathsf{NC}_2\subseteq \mathsf{P}$.} 

Another significant development was achieved by Brakerski, Christiano, Mahadev, Vazirani and Vidick \cite{Brakerski+FOCS18} who proposed, using some techniques from \cite{MahadevFOCS18}, a test of quantumness based on the Learning with Errors ($\LWE{}$) assumption, which states that the learning with error problem (informally, inverting a ``noisy'' system of equations) cannot be solved in polynomial time. (See also \cite{Brakerski+TQC20, Kahanamoku+21} for variants of this test.) They showed that this test can be passed with high probability using a polynomial-time quantum device but cannot be solved by any polynomial-time classical device under the $\LWE{}$ assumption, which is a compelling evidence of the superiority of quantum computing.\footnote{We stress that the quantum protocol that passes the test does not solve the learning with error problem.} A crucial property of this test is that checking if the computational device passes the test (which thus means checking if the computational device is quantum) can be done efficiently --- this property is not known to be true for many other tests from prior works in quantum supremacy (e.g., \cite{Aaronson+STOC11,Aaronson+A14,Aaronson+CCC17,Bouland+18,Bremner+10,Bremner+PRL16,Bremner+17,Fahri+16,Fujii+PRL18,Fujii+16,Morimae+PRL14,Terhal+04}.)
Finally, the test of quantumness from \cite{Brakerski+FOCS18} has another fundamental property: it can be shown that the only way for a computationally bounded quantum prover to pass the test is to prepare precisely the expected quantum state.\footnote{The proof of this statement relies on the (standard) assumption that the learning with error problem is hard for computationally bounded quantum computation as well.} This property makes it possible to control a computationally bounded quantum prover, and has already lead to many cryptographic applications: producing certifiable randomness from a single untrusted (computationally bounded) quantum device \cite{Brakerski+FOCS18}, self-testing of a single quantum device \cite{Metger+21} and device-independent key distribution \cite{Metger+20}.

\subparagraph{Our results.}
In this paper we investigate complexity-theoretic aspects of quantum protocols passing the above test of quantumness based on $\LWE{}$. While the quantum protocol from \cite{Brakerski+FOCS18} can clearly be implemented in polynomial time, and while prior works discussed its practical realization and gave some promising numerical estimates on the number of qubits needed for its implementation (for instance, Ref.~\cite{Brakerski+FOCS18} mentioned 2000 qubits for a protocol providing 50 bits of security), to our knowledge several theoretical aspects, and in particular depth complexity, have not been investigated so far. 

 We first isolate the main computational task solved by a quantum protocol passing the test. This computational problem, which we denote $\state$, is presented in Section~\ref{sec:prep}. Informally, it asks to prepare a quantum superposition of an arbitrary vector $x$ and its shift $x-s$, where $s$ denotes the solution of the ``noisy'' system of linear equations used in the $\LWE{}$ assumption. Our main technical contribution (the formal statement is in Section~\ref{sec:prep}) shows that this problem can be solved by a constant-depth quantum circuit combined with efficient (low-complexity) classical computation:

\begin{theorem}[Informal version]\label{th:main}
The computational task $\state$ can be solved by a constant-depth quantum circuit combined with logarithmic-depth classical computation. 
\end{theorem}

The model of quantum circuits we consider in Theorem \ref{th:main} is described formally in Section~\ref{subsec:prelim-circuits} and is reminiscent of some models used in prior works on measurement-based quantum computing (in particular Refs.~\cite{Broadbent+09,Browne+10}).
The primary motivation for considering this model is as follows: compared with the practical cost of implementing quantum computation, classical computation (and especially low-complexity computation such as logarithmic-depth classical computation) can be considered as a free resource and thus may not be included in the depth complexity. One possible criticism of our model is that the quantum states created by our constant-depth quantum circuits need to be kept coherent while the logarithmic-depth classical computation is performed, which may be an issue since in terms of decoherence waiting is essentially as difficult as performing quantum computation. We can however argue that classical logarithmic-depth classical computation should be implementable significantly faster than logarithmic-depth quantum computation, thus limiting the impact of decoherence. 

As mentioned above, $\state$ is the main computational task used in the test of quantumness based on $\LWE{}$ and its applications given in \cite{Brakerski+FOCS18, Brakerski+TQC20, Metger+20,Metger+21} (the other quantum steps indeed only consist in measuring the state generated in an appropriate basis). As a consequence of Theorem \ref{th:main}, the whole test of quantumness and its applications to producing certifiable randomness, self-testing and device-independent key distribution can thus immediately be implemented by constant-depth quantum circuits combined with logarithmic-depth classical computation. For completeness, we describe in detail how to apply our construction with the whole test of quantumness from \cite{Brakerski+FOCS18}, which was actually only sketched in prior works (since those works focused on applications of the test), in Section~\ref{sec:prior}.

\subparagraph{Overview of our techniques.}
Our main technical contribution is Theorem \ref{th:main}, which shows how to solve $\state$ using constant-depth quantum circuits (in our model allowing some low-complexity classical pre/processing). This is done by modifying the construction of prior works in two major ways.

Our first contribution is to show how to construct in constant depth a quantum state robust against small ``noise''. In~\cite{Brakerski+FOCS18} the construction was done by considering a state with amplitudes taken from a wide-enough Gaussian distribution, and creating this state using the approach from the seminal paper by Regev \cite{Regev09}, which itself relied on a technique by Grover and Rudolph \cite{Grover+02}. To our knowledge, the resulting construction, while definitely implementable with quantum circuits of polynomial size, does not seem to be implementable in constant depth. Instead, our main idea (see Theorem \ref{th:creation} in Section \ref{sec:prep})  is to use a quantum state with amplitudes taken from a much simpler distribution (a wide-enough truncated uniform distribution) that can be implemented in constant depth. 

The second contribution (Theorem \ref{th:preparation} in Section \ref{sec:prep})  is analyzing carefully how to implement in the quantum setting the map used in the learning with error problem (note that in the quantum setting the map needs to be applied in superposition, which requires a quantum circuit). We observe that when given as input a state robust against small noise, the remaining computational task involves only algebraic operations modulo~$q$, for some large integer $q$. We then show that prior works by H{\o}yer and Spalek \cite{Hoyer+05} and Takahashi and Tani~\cite{Takahashi+16} imply that  implementing arithmetic operations modulo~$q$ exactly and generating a good  approximation of the uniform superposition of all elements of $\{0,1,\ldots,q-1\}$ can be done using constant-depth quantum circuits if unbounded fanout gates are allowed. We finally show that unbounded fanout gates can be implemented in our model using a technique called gate teleportation \cite{Gottesman+99,Leung04,Nielsen03}.

\section{Preliminaries}
\subsection{General notations}
In this paper the notation $\log$ represents the logarithm in basis 2. 
For any integer $q$, we write $\Int_q=\{0,1\ldots,q-1\}$.  As usual in lattice-based cryptography, we will often identify $\Int_q$ with the set of integers $\{-\ceil{q/2}+1,\ldots,\floor{q/2}\}$. For any $a\in\Int_q$, we write $J(a)\in\{0,1\}^{\ceil{\log q}}$ its binary representation, as in \cite{Brakerski+FOCS18}.
Given a vector $x\in \Int_q^{m}$, we write $\norm{x}=\sqrt{\sum_{i=1}^m |x_i|^2}$ and $\normmax{x}=\max_{i\in\{1,\ldots,m\}}|x_i|$, and write $J(x)=(J(x_1),\ldots,J(x_m))\in \{0,1\}^{m\ceil{\log q}}$ its binary representation. Given a matrix $A\in \Int_q^{m\times n}$, we define the distance of~$A$ as the minimum over all the non-zero vectors $x\in\Int_q^m$, of the quantity $\norm{Ax}$.

\subsection{Lattice-based cryptography}\label{sub:lattice}
For a security parameter $\lambda$, let $m,n,q$ be integer functions of $\lambda$. Let $\chi$ be a distribution over~$\Int_q$. The $\LWE{m,n,q,\chi}$ problem is to distinguish between the distributions $(A,As+e)$ and $(A,u)$, where $A\in\Int_q^{m\times n}$, $s\in\Int_q^n$ and $u\in\Int_q^m$ are uniformly random and $e\gets \chi^m$. The corresponding hardness assumption is that no polynomial-time algorithm can solve this problem with non-negligible advantage in~$\lambda$.  As in \cite{Brakerski+FOCS18}, we write $\LWE{n,q,\chi}$ the task of solving $\LWE{m,n,q,\chi}$ for any function $m$ that is at most a polynomial in $n\log q$.

The most usual distribution $\chi$ used in lattice-based cryptography is the truncated discrete Gaussian distribution, which we now introduce. For any positive integer $q$ and any positive real number $B$, the truncated discrete Gaussian distribution over $\Int_q$ with parameter $B$, which we denote $D_{q,B}$, is defined as 
$D_{q,B}(x)=(e^{-\pi|x|^2/B^2})/\gamma$ if $|x|\le B$ and $D_{q,B}(x)=0$ otherwise,
%
for any $x\in \Int_q$, where $\gamma$ is the normalization factor defined as $\gamma=\sum_{z\in\Int_q, |z|\le B\:\:}e^{-\pi|z|^2/B^2}$.

As in \cite{Brakerski+FOCS18}, we will use the following theorem to generate instances of the learning with error problem.
\begin{theorem}[Theorem~2.6 in \cite{Brakerski+FOCS18} and Theorem 5.1 in \cite{Micciancio+12}]\label{theorem:crypto}
Let $m,n\ge 1$ and $q\ge 2$ be such that $m=\Omega(n\log q)$. There is an efficient randomized algorithm $\gentrap(1^n,1^m,q)$ that returns a matrix $A\in\Int_q^{m\times n}$ and a trapdoor $t_A$ such that the distribution of $A$ is negligibly (in $n$) close to the uniform distribution. Moreover, there is an efficient algorithm $\invert$ that, on input $A$, $t_A$ and $Ax+e$ where $x\in\Int_q^n$ is arbitrary, $\norm{e}\le q/(C\sqrt{n\log q})$ and $C$ is a universal constant, returns $x$ with overwhelming probability over $(A,t_A)\gets \gentrap(1^n,1^m,q)$.  
\end{theorem}

The matrix $A$ generated by $\gentrap(1^n,1^m,q)$ has distance at least $2q/(C\sqrt{n\log q})$ with overwhelming probability.
Also note that if $\normmax{e}\le q/(C\sqrt{mn\log q})$, then the inequality $\norm{e}\le q/(C\sqrt{n\log q})$ holds. These two observations motivate the following definition: we define $\Kk$ as the set of 5-tuples $(m,n,q,A,u)$ such that $m$, $n$ and $q$ are positive integers, $A\in\Int_q^{m\times n}$ is a matrix of distance at least $2q/(C\sqrt{n\log q})$, where $C$ is the constant from Theorem \ref{theorem:crypto}, and $u\in\Int_q^m$ is a vector that can be written as $u=As+e$ for some $s\in\Int_q^n$ and some $e\in\Int_q^m$ with $\normmax{e}\le q/(C\sqrt{mn\log q})$.
Informally, the set $\Kk$ represents the set of good parameters for the version of $\LWE{}$ we will consider. For technical reasons, we also define the following variant, which enables us to set a stronger upper bound on $\normmax{e}$. For any $B_V> 0$, we define $\Kk_{B_V}\subseteq \Kk$ as the set of 5-tuples $(m,n,q,A,u)\in\Kk$ such that the following two conditions hold:
\begin{itemize}
\item[(i)]
$q\ge B_VC\sqrt{mn\log q}$,
\item[(ii)]
 $u$ can be written as $u=As+e$ for some $s\in\Int_q^n$ and some $e\in\Int_q^m$ with $\normmax{e}\le B_V$.
\end{itemize}

\subsection{Quantum states: bounded and robust states}
We assume that the reader is familiar with the basics of quantum computing and refer to, e.g., \cite{NC00book} for a good reference.

For any positive integer $q$, we write $\Hh_q$ the complex Hilbert space of dimension $q$ with basis $\{\ket{x}\}_{x\in\Int_q}$. 
Quantum states in $\Hh_q$ are (implicitly) implemented using $\ceil{\log q}$ qubits, via the binary encoding of these basis vectors. For any integer $m\ge 1$,  we also consider the Hilbert space $\Hh_q^{\otimes m}$ and associate to it the basis $\{\ket{x}\}_{x\in\Int^m_q}$. A quantum state $\ket{\varphi}$ in $\Hh_q^{\otimes m}$ can thus be written as 
$
\ket{\varphi}=\sum_{x\in\Int_q^m} \alpha_x \ket{x},
$
for complex numbers $\alpha_x$ such that $\sum_{x\in\Int_q^m} |\alpha_x|^2=1$. 
We write its support
$
\supp(\ket{\varphi})=\{x\in \Int_q^m\:|\:\alpha_x\neq 0\}.
$ 
We say that $\ket{\varphi}$ has real amplitudes if $\alpha_x\in\mathbb{R}$ for each $x\in\Int_q^m$. For any vector $e\in\Int_q^m$, we write
$
\ket{\varphi+e}=\sum_{x\in\Int_q^m} \alpha_x \ket{x+e},
$
where the addition is performed modulo $q$.

We now introduce two crucial definitions on which our approach will be based.\footnote{We stress that these two definitions (as well as several definitions of the previous paragraph) are basis-dependent --- we always refer to the canonical basis $\{\ket{x}\}_{x\in\Int^m_q}$. Also note that while Definition \ref{def:rob} can easily be written without the requirement that the state has real amplitude (by replacing $\langle \varphi{\ket{\varphi+e}}$ by $|\langle \varphi{\ket{\varphi+e}}|$, for instance), requiring that the state has real amplitudes will be enough for our purpose and will simplify later calculations.}

\begin{definition} 
Let  $B$ be a positive real number. 
A quantum state $\ket{\varphi}\in\Hh_q^{\otimes m}$ is $B$-bounded if $\normmax{x}< B$ for any element $x\in\supp(\ket{\varphi})$.
\end{definition}

\begin{definition}\label{def:rob}
Let $B$, $\varepsilon$ be two positive real numbers.
A quantum state $\ket{\varphi}\in\Hh_q^{\otimes m}$
is $(\varepsilon,B)$-robust if $\ket{\varphi}$ has real amplitudes and, for any vector $e\in\Int_q^m$ such that $\normmax{e}\le B$, the inequality
$
\langle \varphi{\ket{\varphi+e}}\ge 1-\varepsilon
$
holds.
\end{definition}

Finally, given two states $\ket{\varphi}$ and $\ket{\psi}$ in $\Hh_q^{\otimes m}$, and any positive real number $\varepsilon$, we say that $\ket{\varphi}$ and $\ket{\psi}$ are $\varepsilon$-close if $\norm{\ket{\varphi}-\ket{\psi}}^2\le \varepsilon$.
We also define the notion of $\varepsilon$-closeness to a subspace as follows.
\begin{definition}
Let $\Hh'$ be a subspace of $\Hh_q^{\otimes m}$ and $\varepsilon$ be a positive real number. We say that a state $\ket{\varphi}\in \Hh_q^{\otimes m}$ is $\varepsilon$-close to $\Hh'$ if there exists a quantum state $\ket{\psi}\in\Hh'$ such that $\norm{\ket{\varphi}-\ket{\psi}}^2\le \varepsilon$.
\end{definition}

\subsection{Quantum circuits}\label{subsec:prelim-circuits}
\subparagraph{Universal sets of quantum gates.}
As in the standard model of quantum circuits (see, e.g.,~\cite{NC00book}), in this paper we work with qubits. We consider two sets of elementary gates.  We first consider the set $\basis_\mathrm{r}=\{H, T, CNOT\}$ where $H=\frac{1}{\sqrt{2}}\begin{psmallmatrix}1&1\\
1&-1\end{psmallmatrix}$ is the Hadamard gate, $T=\begin{psmallmatrix}1&0\\
0&e^{i\pi/4}\end{psmallmatrix}$ is the $\pi/8$-phase operation and $CNOT=\begin{psmallmatrix}1&0&0&0\\
0&1&0&0\\
0&0&0&1\\
0&0&1&0\end{psmallmatrix}$ is the controlled-not gate.
%
This is an universal set consisting of a finite number of gates that can approximate any quantum gate with good precision (see Section 4.5.3 of \cite{NC00book} for details). 
The second set we consider, which we denote $\Bb$, contains all the gates acting on 1 qubit and the $CNOT$ operator. Note that this set contains an infinite number of gates.

\subparagraph{Our model.}
We now introduce the class of quantum circuits considered in this paper. Let $r_1$ and $r_2$ be two positive integers, and $\Ss$ be a set of elementary quantum gates (e.g., $\Ss=\Bb_{\mathrm{r}}$ or $\Ss=\Bb$).

A circuit in the class $\Cc(\Ss,r_1,r_2)$ acts on $r_1+r_2$ qubits. These qubits are initialized to the state $\ket{0}^{\otimes (r_1+r_2)}$. The circuit consists of successive layers. 
Each layer consists of a constant-depth quantum circuit over the basis $\Ss$ acting on these $r_1+r_2$ qubits, which does not contains any measurement, followed by measurements in the computational basis of all the first $r_1$ qubits. Consider the $i$-th layer. Let $x_i\in\{0,1\}^{r_1}$ denote the outcome of measuring the first $r_1$ qubits at the end of this layer. Then some classical function $f_i\colon\{0,1\}^{r_1}\to\{0,1\}^{r_1}$ is applied to the $x_i$, and the value $f_i(x_i)$ is given as input to the first $r_1$ qubits of the next layer, i.e., the $r_1$ qubits are reinitialized to the state $\ket{f_i(x_i)}$. We refer to Figure \ref{fig:circuits} for an illustration.

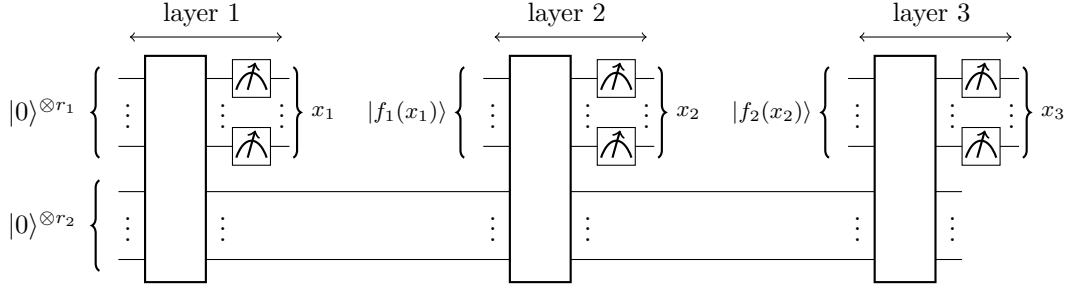
\begin{figure}[ht!]
\centering
\begin{tikzpicture}[scale=0.5,rectnode/.style={thick,shape=rectangle,draw=black,minimum height=30mm, minimum width=8mm},rectnode2/.style={shape=rectangle,draw=black,minimum height=5mm, minimum width=5mm},roundnode/.style={circle, draw=green!60, fill=green!5, very thick, minimum size=7mm}]
    \newcommand\XA{9.6}
    \newcommand\YA{6}
        
    \node[rectnode] (a2) at (0*\XA,1*\YA) {};
    
    \node[rectnode] (b2) at (1*\XA,1*\YA) {};
    
    \node[rectnode] (c2) at (2*\XA,1*\YA) {};


     \draw (-1.5,1.4*\YA) -- (-0.8,1.4*\YA);
     \draw (-1.5,1.1*\YA) -- (-0.8,1.1*\YA);
     \draw (-1.5,0.9*\YA) -- (-0.8,0.9*\YA);
     \draw (-1.5,0.6*\YA) -- (-0.8,0.6*\YA);

     \draw (0.8,0.9*\YA) -- (8.8,0.9*\YA);
     \draw (0.8,0.6*\YA) -- (8.8,0.6*\YA);

     \draw (0.8+\XA,0.9*\YA) -- (8.8+\XA,0.9*\YA);
     \draw (0.8+\XA,0.6*\YA) -- (8.8+\XA,0.6*\YA);

	 \draw (-1.5+\XA,1.4*\YA) -- (-0.8+\XA,1.4*\YA);
     \draw (-1.5+\XA,1.1*\YA) -- (-0.8+\XA,1.1*\YA);
 	 \draw (-1.5+2*\XA,1.4*\YA) -- (-0.8+2*\XA,1.4*\YA);
     \draw (-1.5+2*\XA,1.1*\YA) -- (-0.8+2*\XA,1.1*\YA);
     
     \draw (-1.5+2.3,1.4*\YA) -- (-0.8+2.3,1.4*\YA);
     \draw (-1.5+2.3,1.1*\YA) -- (-0.8+2.3,1.1*\YA);
	 \draw (-1.5+2.3+\XA,1.4*\YA) -- (-0.8+2.3+\XA,1.4*\YA);
     \draw (-1.5+2.3+\XA,1.1*\YA) -- (-0.8+2.3+\XA,1.1*\YA);
     \draw (-1.5+2.3+2*\XA,1.4*\YA) -- (-0.8+2.3+2*\XA,1.4*\YA);
     \draw (-1.5+2.3+2*\XA,1.1*\YA) -- (-0.8+2.3+2*\XA,1.1*\YA);

     \draw (8*\XA+0.25,1*\YA+0) -- (8*\XA+1,1*\YA+0);
     \draw (8*\XA+0.25,1*\YA-0.35) -- (8*\XA+1,1*\YA-0.35);
     \draw (-1.5+2.3+2*\XA,0.9*\YA) -- (-0.8+2.3+2*\XA,0.9*\YA);
     \draw (-1.5+2.3+2*\XA,0.6*\YA) -- (-0.8+2.3+2*\XA,0.6*\YA);

     \node[draw=none,fill=none] at (0.7,1.3*\YA+2.3) {layer 1};
     \node[draw=none,fill=none] at (0.7+\XA,1.3*\YA+2.3) {layer 2};
     \node[draw=none,fill=none] at (0.7+2*\XA,1.3*\YA+2.3) {layer 3};
     \draw[<->] (-1.2,1.3*\YA+1.7) -- (2.8,1.3*\YA+1.7);
     \draw[<->] (-1.2+\XA,1.3*\YA+1.7) -- (2.8+\XA,1.3*\YA+1.7);
     \draw[<->] (-1.2+2*\XA,1.3*\YA+1.7) -- (2.8+2*\XA,1.3*\YA+1.7);
     
     \draw[thick,decorate,decoration={brace,amplitude=1mm}] (-2,-1.2+0.75*\YA) -- (-2,1.2+0.75*\YA);
     \node[draw=none,fill=none] at (-3.5,1.25*\YA) {$\ket{0}^{\otimes r_1}$};
     \node[draw=none,fill=none] at (-1.25,1.27*\YA) {$\vdots$};
     \node[draw=none,fill=none] at (1.25,1.27*\YA) {$\vdots$};
     \node[draw=none,fill=none] at (-1.25+\XA,1.27*\YA) {$\vdots$};
     \node[draw=none,fill=none] at (1.25+\XA,1.27*\YA) {$\vdots$};
     \node[draw=none,fill=none] at (-1.25+2*\XA,1.27*\YA) {$\vdots$};
     \node[draw=none,fill=none] at (1.25+2*\XA,1.27*\YA) {$\vdots$};
     \draw[thick,decorate,decoration={brace,amplitude=1mm}] (-2,-1.2+1.25*\YA) -- (-2,1.2+1.25*\YA);
     \node[draw=none,fill=none] at (-3.5,0.75*\YA) {$\ket{0}^{\otimes r_2}$}; 
     \node[draw=none,fill=none] at (-1.25,0.77*\YA) {$\vdots$};
     \node[draw=none,fill=none] at (-1.25+2*\XA,0.77*\YA) {$\vdots$};
      \node[draw=none,fill=none] at (1.25+2*\XA,0.77*\YA) {$\vdots$};

     \node[draw=none,fill=none] at (1.25+\XA,0.77*\YA) {$\vdots$};
     \node[draw=none,fill=none] at (-1.25+\XA,0.77*\YA) {$\vdots$};
 \node[draw=none,fill=none] at (1.25+0*\XA,0.77*\YA) {$\vdots$};
     \draw[thick,decorate,decoration={brace,amplitude=1mm}] (-2+\XA,-1.2+1.25*\YA) -- (-2+\XA,1.2+1.25*\YA);
     \draw[thick,decorate,decoration={brace,amplitude=1mm}] (-2+2*\XA,-1.2+1.25*\YA) -- (-2+2*\XA,1.2+1.25*\YA);    
      \node[draw=none,fill=none] at (-3.5+\XA,1.25*\YA) {\footnotesize$\ket{f_1(x_1)}$};
	 \node[draw=none,fill=none] at (-3.5+2*\XA,1.25*\YA) {\footnotesize$\ket{f_2(x_2)}$};
    
     \draw[thick,decorate,decoration={brace,mirror,amplitude=1mm}] (3.1+0*\XA,-1.2+1.25*\YA) -- (3.1+0*\XA,1.2+1.25*\YA); 
     \node[rectnode2] (b2) at (0*\XA+2,1.4*\YA) {};
     \node[rectnode2] (b2) at (0*\XA+2,1.1*\YA) {};
     \draw (-1.5+4,1.4*\YA) -- (-1+4,1.4*\YA);
     \draw (-1.5+4,1.1*\YA) -- (-1+4,1.1*\YA);
     \node[draw=none,fill=none] at (2.8,1.27*\YA) {$\vdots$};
      \node[draw=none,fill=none] at (-5.7+\XA,1.25*\YA) {\footnotesize $x_1$};

     \draw[thick,decorate,decoration={brace,mirror,amplitude=1mm}] (3.1+1*\XA,-1.2+1.25*\YA) -- (3.1+1*\XA,1.2+1.25*\YA); 
     \node[rectnode2] (b2) at (1*\XA+2,1.4*\YA) {};
     \node[rectnode2] (b2) at (1*\XA+2,1.1*\YA) {};
     \draw (-1.5+\XA+4,1.4*\YA) -- (-1+4+\XA,1.4*\YA);
     \draw (-1.5+\XA+4,1.1*\YA) -- (-1+4+\XA,1.1*\YA);
     \node[draw=none,fill=none] at (2.8+\XA,1.27*\YA) {$\vdots$};
      \node[draw=none,fill=none] at (-5.7+2*\XA,1.25*\YA) {\footnotesize$x_2$};

     \draw[->, thick] (0*\XA+1.9,1.35*\YA) -- (0*\XA+2.1,1.47*\YA);
     \draw[thick] (0*\XA+1.6,1.35*\YA) .. controls  (0*\XA+1.8,1.45*\YA) and  (0*\XA+2.2,1.45*\YA) .. (0*\XA+2.4,1.35*\YA);
     \draw[->, thick] (0*\XA+1.9,1.05*\YA) -- (0*\XA+2.1,1.17*\YA);
     \draw[thick] (0*\XA+1.6,1.05*\YA) .. controls  (0*\XA+1.8,1.15*\YA) and  (0*\XA+2.2,1.15*\YA) .. (0*\XA+2.4,1.05*\YA);

     \draw[->, thick] (1*\XA+1.9,1.35*\YA) -- (1*\XA+2.1,1.47*\YA);
     \draw[thick] (1*\XA+1.6,1.35*\YA) .. controls  (1*\XA+1.8,1.45*\YA) and  (1*\XA+2.2,1.45*\YA) .. (1*\XA+2.4,1.35*\YA);
     \draw[->, thick] (1*\XA+1.9,1.05*\YA) -- (1*\XA+2.1,1.17*\YA);
     \draw[thick] (1*\XA+1.6,1.05*\YA) .. controls  (1*\XA+1.8,1.15*\YA) and  (1*\XA+2.2,1.15*\YA) .. (1*\XA+2.4,1.05*\YA);

    \draw[thick,decorate,decoration={brace,mirror,amplitude=1mm}] (3.1+2*\XA,-1.2+1.25*\YA) -- (3.1+2*\XA,1.2+1.25*\YA); 
     \node[rectnode2] (b2) at (2*\XA+2,1.4*\YA) {};
     \node[rectnode2] (b2) at (2*\XA+2,1.1*\YA) {};
     \draw (-1.5+2*\XA+4,1.4*\YA) -- (-1+4+2*\XA,1.4*\YA);
     \draw (-1.5+2*\XA+4,1.1*\YA) -- (-1+4+2*\XA,1.1*\YA);
     \node[draw=none,fill=none] at (2.8+2*\XA,1.27*\YA) {$\vdots$};
      \node[draw=none,fill=none] at (-5.7+3*\XA,1.25*\YA) {\footnotesize $x_3$};

     \draw[->, thick] (2*\XA+1.9,1.35*\YA) -- (2*\XA+2.1,1.47*\YA);
     \draw[thick] (2*\XA+1.6,1.35*\YA) .. controls  (2*\XA+1.8,1.45*\YA) and  (2*\XA+2.2,1.45*\YA) .. (2*\XA+2.4,1.35*\YA);
     \draw[->, thick] (2*\XA+1.9,1.05*\YA) -- (2*\XA+2.1,1.17*\YA);
     \draw[thick] (2*\XA+1.6,1.05*\YA) .. controls  (2*\XA+1.8,1.15*\YA) and  (2*\XA+2.2,1.15*\YA) .. (2*\XA+2.4,1.05*\YA);
\end{tikzpicture}
\caption{A quantum circuit of the class $\Cc$ consisting of three layers. Each rectangular box represents a quantum circuit (without measurements) of constant depth with gates in the set $\Ss$.}\label{fig:circuits}
\end{figure}

The complexity of a circuit in the class defined above depends on the number of qubits $r_1+r_2$, the number of layers and the classical complexity of computing function $f_i$'s. We are mainly interested in circuits that have a constant number of layers and such that all functions can be computed efficiently classically. We formally define this class below.  

We define the class $\Cc_0(\Ss)$ of families of circuits $\{C_n\}_{n\in\mathbb{N}}$ such that the following conditions hold:
\begin{itemize}
\item
for each $n\in \mathbb{N}$, we have $C_n\in \Cc(\Ss,r_1,r_2)$ for some integers $r_1,r_2$ such that $r_1+r_2=n$;
\item
for each $n\in \mathbb{N}$, the number of layers in $C_n$ is constant (i.e., independent of $n$);
\item
for each $n\in \mathbb{N}$,
all the functions $f_i$'s of $C_n$ can be computed by a $O(\log n)$-depth classical circuit.
\end{itemize}
We require that the family is logarithmic-space uniform, i.e., there exists a classical Turing machine that on input $1^n$ outputs a classical description of $C_n$ (as well as descriptions of the circuits computing the functions $f_i$'s) in $O(\log n)$ space.

\subsection{Clifford circuits and quantum arithmetic}\label{sub:clifford}
\subparagraph{Clifford circuits.}
Let us consider the Pauli gates
$X=
\begin{psmallmatrix}
0&1\\
1&0
\end{psmallmatrix}
$
and
$Z=
\begin{psmallmatrix}
1&0\\
0&-1
\end{psmallmatrix}
$ and the phase gate $S=
\begin{psmallmatrix}
1&0\\
0&i
\end{psmallmatrix}$.
A quantum circuit consisting only of gates from the set~$\{X,Z,S, H,CNOT\}$ is called a Clifford circuit.\footnote{Since $X=S^2$ and $Z=HS^2H$, the two Pauli gates can actually be removed from this gate set.} 
Such a circuit can be implemented by a quantum circuit of class $\Cc_0(\Bb_\mathrm{r})$ acting on $\poly(s)$ qubits, where $s$ is the number of gates in the original circuit, via a technique called gate teleportation first introduced by Gottesman and Chuang~\cite{Gottesman+99} and then developed into a computational model by Leung \cite{Leung04} and Nielsen~\cite{Nielsen03}
(see also, e.g.,~\cite{Brakerski+20,Jozsa05} for good presentations of this technique). 

A concrete example, which we will actually heavily use, is the unbounded fanout gate over $\Hh_2^{\otimes m}$. This unitary gate
maps the basis state $\ket{x_1,x_2,\ldots,x_{m-1},x_m}$ to $\ket{x_1,x_1\oplus x_2,\ldots,x_{1}\oplus x_{m-1},x_1\oplus x_m}$, for any $x_1,\ldots,x_m\in\{0,1\}$. This gate can easily be written as a circuit consisting of $m-1$ successive CNOT gates (the depth of such a circuit implementation is thus linear in $m$). Using the above approach, this gate can be implemented by a quantum circuit of class $\Cc_0(\Bb_\mathrm{r})$ acting on $\poly(m)$ qubits. A concrete decomposition, which uses only two layers, is presented in Section~6 of \cite{Broadbent+09}.

\subparagraph{Modular arithmetics.}
Let us consider the following unitary operations (where the arithmetic operations are performed modulo $q$ and $\omega$ is a $q$-th root of unity): 
\begin{itemize}
\item
the quantum Fourier transform $F_q$ over $\Hh_q$, such that 
$
F_q\ket{i}=\frac{1}{\sqrt{q}}\sum_{j=0}^{q-1}\omega^{ij}\ket{j}
$
for any $i\in \Int_q$;
\item
the unitary operation $\ADD_q$ over $\Hh_q^{\otimes 2}$ that maps $\ket{i}\ket{j}$ to $\ket{i}\ket{i+j}$ for any $i,j\in \Int_q$;
\item
the unitary operation $\MULT_q$ over $\Hh_q^{\otimes 3}$ that maps $\ket{i}\ket{j}\ket{k}$ to $\ket{i}\ket{j}\ket{k+ij}$ for any $i,j,k\in \Int_q$.
\end{itemize}

We now discuss how to obtain exact implementations for $\ADD_q$ and $\MULT_q$, and also for arbitrary linear maps over $\Int_q$ (exact implementation of these gates will be crucial for implementing our test of quantumness in constant depth).  Takahashi and Tani \cite{Takahashi+16} showed how to implement exactly $\ADD_q$ and $\MULT_q$ in constant depth by circuits that use gates in~$\basis$ and unbounded fanout gates acting on $\poly(\log q)$ qubits, by showing that quantum threshold gates, which are enough to implement all these operations (as first pointed out by H{\o}yer and Spalek \cite{Hoyer+05}, based on prior works on classical threshold gates \cite{Siu+93}), can be implemented in constant depth by such circuits. Since each unbounded fanout gate can be implemented by a quantum circuit of class $\Cc_0(\Bb_\mathrm{r})$ acting on $\poly(\log q)$ qubits, as discussed above, these arithmetic operations can be exactly implemented by quantum circuits of class $\Cc_0(\Bb)$ acting on $\poly(\log q)$ qubits. As discussed in \cite{Hoyer+05,Takahashi+16}, the same approach can be applied to implement iterated addition, and more generally any linear map $f\colon \Int_q^n\to \Int_q$, since such maps can be computed in constant depth using classical threshold gates as well. This can easily be further generalized to give implementation of any linear map $f\colon \Int_q^n\to \Int_q^m$ by a quantum circuit of class $\Cc_0(\Bb)$ acting on $\poly(m,n,\log q)$ qubits.

Unfortunately, it is still unknown if the operator $F_q$ can be implemented exactly in constant depth with a circuit using only elementary gates in $\Bb$ and unbounded fanout gates (see Section 6 of~\cite{Takahashi+16}). For the protocol constructed in this paper, however, we will only need to apply $F_q$ to the state $\ket{0}\in\Hh_q$, i.e., we only need to prepare the state $F_q\ket{0}=\frac{1}{\sqrt{q}}\sum_{x\in\Int_q}\ket{x}$. Lemma~4.18 in \cite{Hoyer+05} shows that this task can be implemented in constant depth with exponential precision (which will be enough for our purpose): there exists a constant-depth circuit of size $\poly(\log q)$ using gates in $\Bb$ and unbounded fanout gates that computes a state which is at distance at most $1/q^2$ of the state $F_q\ket{0}$. By converting each unbounded fanout gate, this circuit can immediately be converted into a circuit in the class $\Cc_0(\Bb)$ acting on $\poly(\log q)$ qubits.

\section{Quantum State Generation using Small-Depth Circuits}\label{sec:prep}
In this section we describe the main computational task solved by a quantum prover in the test of quantumness based on $\LWE{}$ we present in Section \ref{sec:prior} (as well as in prior works \cite{Brakerski+FOCS18, Brakerski+TQC20, Metger+20,Metger+21}), and show how to solve it using a quantum circuit of small depth.

\subsection{Statement of the problem}
For any $B_V> 0$ and any $k=(m,n,q,A,u)\in\Kk_{B_V}$, where $\Kk_{B_V}$ is the set of parameters defined in Section~\ref{sub:lattice}, let $\Lambda_k\subseteq\Int_q^m$ denote the set of vectors $y\in\Int_q^m$ such that there exists a vector $x\in\Int_q^n$ for which  $\norm{Ax-y}\le q/(C\sqrt{n\log q})$. Note that such $x$ is necessarily unique, since $A$ has distance at least $2q/(C\sqrt{n\log q})$. Let us write this vector $x_y$. Note that $x_u=s$ using the notations of Section \ref{sub:lattice}, i.e., defining $s$ as the (unique) vector in $\Int_q^n$ such that $u$ can be written as $u=As+e$ for $e\in\Int_q^m$ with $\normmax{e}\le B_V$. For any $y\in\Lambda_k$, define the quantum state 
\[
\ket{\Psi_y} = \frac{1}{\sqrt{2}}\left(\ket{0}\ket{x_y}+\ket{1}\ket{x_y-x_u}\right).
\]
Let $\Hhh_k$ be the subspace of $\Hh_2\otimes\Hh_q^n\otimes \Hh_q^m$ generated by the states $\{\ket{\Psi_y}\ket{y}\}_{y\in\Lambda_k}$.

The computational problem we consider in this section, which we denote $\state$, has two parameters $\varepsilon,B_V>0$, and is defined as follows. This is the main task solved by the quantum protocols passing our test of quantumness, as well as in the tests used in prior works \cite{Brakerski+FOCS18, Brakerski+TQC20, Metger+20,Metger+21}.\vspace{2mm}

$\state(\varepsilon,B_V)$

\framebox[1.05\width]{

Given $k\in\Kk_{B_V}$, create a quantum state $\varepsilon$-close to $\Hhh_k$.
}\vspace{3mm}

Here is our main theorem, which shows that the problem can be solved by a small-depth quantum circuit when $q$ is large enough.

\addtocounter{theorem}{-5}
\begin{theorem}[Formal version]
For any $\varepsilon,B_V> 0$, the problem $\state(\varepsilon,B_V)$ can be solved, for all inputs $k\in\Kk_{B_V}$ such that $q\ge (8mB_VC\sqrt{mn\log q})/ \varepsilon$, by a quantum circuit of class $\Cc_0(\Bb)$ acting on $\poly(m,n,\log q)$ qubits.
\end{theorem}
\addtocounter{theorem}{4}
Theorem \ref{th:main} follows from Theorems \ref{th:preparation} and \ref{th:creation} proved in Subsections \ref{sub:preparation} and \ref{sub:creation}.

%
\subsection{Preparation procedure}\label{sub:preparation}
In this subsection we present and analyze a quantum procedure 
that outputs a state close to $\Hhh_k$ when given as additional input an appropriate quantum state $\ket{\varphi}\in\Hh_q^{\otimes m}$. This procedure can be implemented by a small-depth quantum circuit.
In subsection \ref{sub:creation} we will show how to create efficiently such an appropriate state $\ket{\varphi}$.

The following theorem is the main contribution of this subsection.
\begin{theorem}\label{th:preparation}
Let $\varepsilon$ and $B_V$ be any positive parameters. For any $k\in \Kk_{B_V}$ with $q\ge \sqrt{2n/\varepsilon}$, there exists a quantum circuit of class $\Cc_0(\Bb)$ acting on $\poly(m,n,\log q)$ qubits that receives a quantum state $\ket{\varphi}\in\Hh_q^{\otimes m}$, outputs a quantum state $\ket{\Phi}\in\Hh_2\otimes\Hh_q^n\times \Hh_q^m$, and satisfies the following condition: if  $\ket{\varphi}$ is $\frac{q}{C\sqrt{mn\log q}}$-bounded and $(\varepsilon/2,B_V)$-robust, then 
$\ket{\Phi}$ is $\varepsilon$-close to~$\Hhh_k$.
\end{theorem}
\begin{proof}
We first describe the procedure. Let us write
$
\ket{\varphi}=\sum_{z\in \Int_q^m}\alpha_z\ket{z}
$
the input state, where $\alpha_z\in\mathbb{R}$ for all $z\in\Int_q^m$ (remember that the definition of a robust state implies that the amplitudes are real).
The procedure first prepares the state $\ket{0}\ket{0}\ket{\varphi}\in\Hh_2\otimes\Hh_q^{\otimes n}\otimes \Hh_q^{\otimes m}$ and applies the unitary operator $H\otimes F_q^{\otimes n}\otimes I$ to this state to obtain
\[
\frac{1}{\sqrt{2q^n}}\sum_{b\in\{0,1\}}\sum_{x\in\Int_q^n}\ket{b}\ket{x}\ket{\varphi}
=
\sum_{b\in\{0,1\}}
\sum_{x\in\Int_q^n}\sum_{z\in \Int_q^m}\frac{\alpha_z}{\sqrt{2q^n}}\ket{b}\ket{x}\ket{z}.
\]
Using the approach discussed in Section~\ref{sub:clifford}, this can be done by a quantum circuit of class $\Cc_0(\Bb)$ acting on $\poly(m,n,\log q)$ qubits with approximation error $\frac{n}{q^2}\le \varepsilon/2$. Below we assume that this state has been done exactly --- we will add the approximation error at the very end of the calculation. 
The procedure then converts this state to the state
\[
\ket{\Phi}=\sum_{b\in\{0,1\}}\sum_{x\in\Int_q^n}\sum_{z\in \Int_q^m}\frac{\alpha_z}{\sqrt{2q^n}}
\ket{b}\ket{x}\ket{z+f_k(b,x)},
\]
where $f_{k}\colon \{0,1\}\times \Int_q^n\to\Int_q^m$ is the function defined  as
$
f_{k}(b,x)=Ax+bu
$
for any $(b,x)\in \{0,1\}\times \Int_q^n$ (all the operations are performed modulo $q$).
This operation can be implemented by a quantum circuit of class $\Cc_0(\Bb)$ acting on $\poly(m,n,\log q)$ qubits using the approach of Section \ref{sub:clifford} since~$f_k$ can be written as a linear map over $\Int_q\times \Int_q^n$ as follows: define the matrix $A'\in\Int_q^{m \times (n+1)}$ obtained by appending the vector $u$ to the left of the matrix $A$ and write
$
f_{k}(b,x)=A'
\begin{psmallmatrix}
b\\x
\end{psmallmatrix}.
$

We now analyze this procedure. Let us write the output state in the following form:
\begin{align*}
\ket{\Phi}&=\frac{1}{\sqrt{2q^n}}\sum_{x\in\Int_q^n}(\ket{0}\ket{x}\ket{\Phi_{0,x}}+\ket{1}\ket{x}\ket{\Phi_{1,x}}),
\end{align*}
where
\[
\ket{\Phi_{0,x}}=\sum_{z\in \Int_q^m}\alpha_z\ket{Ax+z}
\hspace{3mm}\textrm{and}\hspace{3mm}
\ket{\Phi_{1,x}}=\sum_{z\in \Int_q^m}\alpha_z\ket{Ax+u+z}=\sum_{z\in \Int_q^m}\alpha_z\ket{A(x+s)+e+z}.
\]
Define the quantum state
\begin{align*}
\ket{\Phi'}&=\frac{1}{\sqrt{2q^n}}\sum_{x\in\Int_q^n}(\ket{0}\ket{x}\ket{\Phi'_{0,x}}+\ket{1}\ket{x}\ket{\Phi'_{1,x}}),
\end{align*}
where
$
\ket{\Phi'_{0,x}}=\ket{\Phi_{0,x}}
$
and 
$
\ket{\Phi'_{1,x}}=\sum_{z\in \Int_q^m}\alpha_z\ket{A(x+s)+z}. 
$
We first show that the states $\ket{\Phi}$ and $\ket{\Phi'}$ are close.
\begin{claim}\label{claim1}
$\langle \Phi \ket{\Phi'}\ge 1-\varepsilon/4$.
\end{claim}
\begin{proof}
We have $u=As+e$ for some $s\in\Int_q^n$ and some vector $e\in \Int_q^m$ such that $\normmax{e}\le B_V$. Since the state $\ket{\varphi}$ is $(\varepsilon/2,B_V)$-robust, we thus have
$
\langle \Phi_{1,x}\ket{\Phi'_{1,x}}=\langle \varphi {\ket{\varphi+e}}\ge 1-\varepsilon/2
$
for any $x\in \Int_q^n$. We thus obtain
$
\langle \Phi\ket{\Phi'}=\frac{1}{2}+\frac{1}{2q^n}\sum_{x\in\Int_q^n}\langle\Phi_{1,x}\ket{\Phi'_{1,x}}\ge 1-\varepsilon/4,
$
as claimed.
\end{proof}

We now show that the state $\ket{\Phi'}$ is in $\Hhh_k$.
The crucial property we will use is that the equality
$
\ket{\Phi'_{0,x}}=\ket{\Phi'_{1,x-s}}
$
holds for any $x\in\Int_q^s$.

 Let us decompose $\ket{\Phi'}$ as follows:
\begin{align*}
\ket{\Phi'}&=\sum_{y\in\Int_q^m} \gamma_y \ket{\Phi'_y}\ket{y},
\end{align*}
for quantum states $\ket{\Phi'_y}$ and amplitudes $\gamma_y$ such that 
$
\sum_{y\in \Int_q^m}|\gamma_y|^2=1.
$
We now show the following claim.
\begin{claim}\label{claim3}
For any $y\in \Int_q^m$ such that $|\gamma_y|>0$, we have $y\in\Lambda_k$ and
$\ket{\Phi'_y}=\ket{\Psi_y}$.
\end{claim}
\begin{proof}
Assume that $|\gamma_y|>0$.
Observe that 
in this case
 $y\in\supp(\ket{\Phi'_{0,x_0}})$ for some $x_0\in\Int_q^n$. Since the state $\ket{\varphi}$ is $q/(C\sqrt{mn\log q})$-bounded, we have
$
\norm{y-Ax_0}\le \sqrt{m}\cdot\normmax{y-Ax_0}\le  q/(C\sqrt{n\log q}),
$
and thus $y\in\Lambda_k$.

We show below that
for any distinct $x,x'\in \Int_q^n$ we have
$
\supp(\ket{\Phi'_{0,x}})\cap \supp(\ket{\Phi'_{0,x'}})=\emptyset,
$
which implies that $\ket{\Phi'_y}=\ket{\Psi_y}$.

Indeed, assume that $\supp(\ket{\Phi'_{0,x}})\cap \supp(\ket{\Phi'_{0,x'}})\neq\emptyset$ and take an element $r$ in the intersection. Since the state $\ket{\varphi}$ is $B_P$-bounded, we have
$\norm{r-Ax}\le \sqrt{m}\cdot\normmax{r-Ax}\le q/(C\sqrt{n\log q})$ and $
\norm{r-Ax'}\le \sqrt{m}\cdot\normmax{r-Ax'}\le q/(C\sqrt{n\log q})$,
and thus $\norm{A(x-x')}\le 2q/(C\sqrt{n\log q})$. This is impossible, since by construction the matrix $A$ has distance at least $2q/(C\sqrt{n\log q})$.
\end{proof}
%
%
%
Claim \ref{claim3} implies that the state $\ket{\Phi'}$ is in $\Hhh_k$. Since we have 
$\norm{\ket{\Phi} - \ket{\Phi'}}^2=2-2\langle \Phi \ket{\Phi'}\le \varepsilon/2$
from Claim \ref{claim1}, 
this concludes the proof of the theorem (the additional $\varepsilon/2$ term comes from the approximation error in the application of $F_q^{\otimes n}$).
\end{proof}

\subsection{Creating the initial state}\label{sub:creation}
Brakerski et al.~\cite{Brakerski+FOCS18} have shown how to construct a quantum state that is $B_P$-bounded and  $(\varepsilon,B_V)$-robust, for appropriate parameters $B_V\ll B_P$, using Gaussian distributions. In this subsection we present another quantum state that has similar properties, but can be created by a small-depth quantum circuit.

\begin{theorem}\label{th:creation}
For any $\varepsilon,B_V>0$, any integer $m\ge 1$ and any $q\ge (8mB_VC\sqrt{mn\log q})/ \varepsilon$, there exists a quantum circuit of class $\Cc_0(\Bb)$ acting on $\poly(m,\log q)$ qubits that generates a quantum state $\ket{\varphi}\in\Hh_q^{\otimes m}$ that is $\frac{q}{C\sqrt{mn\log q}}$-bounded and $(\varepsilon/2,B_V)$-robust.
\end{theorem} 
\begin{proof}
Let us write $r=\floor{\log_2 \left(\frac{q}{C\sqrt{mn\log q}}\right)}$ and $I=\{-2^{r-1},\ldots,0,\ldots,2^{r-1}-1\}$.

We describe the construction. Starting with the quantum state $\ket{0}^{\otimes m}\in\Hh_q^{\otimes m}$, apply (in parallel) a Hadamard gate on the first $r$ qubits of each copy of $\ket{0}$, in order to get the state
\[
\Bigg(\frac{1}{\sqrt{2^r}}\sum_{x\in\{0,\ldots,2^r-1\}}\ket{x}\Bigg)^{\otimes m}.
\]
Then apply on each of the $m$ copies the unitary operator over $\Hh_q$ that maps $\ket{i}$ to $\ket{i-2^{r-1}}$ for any $i\in\Int_q$ (the subtraction is done modulo $q$). As described in Section~\ref{sub:clifford}, these arithmetic operations can be implemented by a quantum circuit of class $\Cc_0(\Bb)$ acting on $\poly(m,\log q)$ qubits. This gives the state  
\[
\left(\frac{1}{\sqrt{2^r}}\sum_{x\in I}\ket{x}\right)^{\otimes m}=
\frac{1}{\sqrt{2^{mr}}} \sum_{(x_1,\ldots,x_m)\in I^m}\ket{x_1,\ldots,x_m}.
\]

For any vector $e=(e_1,\ldots,e_m)\in\Int_q^m$, consider the state 
\[
\ket{\varphi+e}=\frac{1}{\sqrt{2^{mr}}} \sum_{(x_1,\ldots,x_m)\in I^m}\ket{x_1+e_1,\ldots,x_m+e_m}.
\]
The inner product of $\ket{\varphi}$ and $\ket{\varphi+e}$ is 
$
\langle \varphi\ket{\varphi+e}=\frac{|S_e|}{2^{mr}},
$
where $S_e$ is the set of vectors $(x_1,\ldots,x_m)\in I^m$ such that $(x_1+e_1,\ldots,x_m+e_m)\in I^m$. If $\normmax{e}\le B_V$, then
$
 \{-2^{r-1}+B_V,\ldots,2^{r+1}-1-B_V\}^m \subset S_e
$
and thus
\[
\langle \varphi\ket{\varphi+e}\ge \left(\frac{2^r-2B_V}{2^r}\right)^m\
=
\left(1-\frac{B_V}{2^{r-1}}\right)^m
\ge 1-\frac{mB_V}{2^{r-1}}\ge 
1-\frac{4mB_VC\sqrt{mn\log q}}{q}
\ge 1-\varepsilon/2,
\]
as claimed.
\end{proof}
\section{Application: Test of Quantumness}\label{sec:prior}
In this section we describe and analyze the test of quantumness based on the $\LWE{}$ assumption that has been implicitly presented in \cite{Brakerski+FOCS18}, and show how to use the results from Section \ref{sec:prep} to pass this test with small-depth quantum circuits.


We first define some sets $G_{s,b,x}\subseteq \{0,1\}^{n\ceil{\log q}}$ exactly as in \cite{Brakerski+FOCS18}. The definition is fairly technical and can actually be skipped on a first reading, since we will later only use the property that these sets are dense enough.
For any $b\in\{0,1\}$ and any $x\in\Int_q^n$, let $I_{b,x}\colon \{0,1\}^{n\ceil{\log q}}\to \{0,1\}^n$ be the map such that for any $d\in \{0,1\}^{n\ceil{\log q}}$, each coordinate of $I_{b,x}(d)$ is obtained by taking the inner product modulo 2 of the corresponding block of $\ceil{\log q}$ coordinates of $d$ and of $J(x)\oplus J(x-(-1)^b\mathbf{1})$, where $\mathbf{1}$ denotes the vector in~$\Int_q^n$ where each coordinate is $1\in\Int_q$. We define the set 
\[
G_{b,x}=\Big\{d\in\{0,1\}^{n\ceil{\log q}}\:|\: \exists i\in\left\{b\frac{n}{2},\ldots,b\frac{n}{2}+\frac{n}{2}\right\}: (I_{b,x}(d))_i\neq 0\}\Big\}.
\]
For any $s\in\Int_q^m$, we then define
$
G_{s,0,x}=G_{0,x}\cap G_{1,x- s}
$
and
$
G_{s,1,x}=G_{0,x+s}\cap G_{1,x}.
$
Note that these sets are dense:
for any $s,x\in\Int_q^n$ and any $b\in\{0,1\}$, we have
 $|G_{s,b,x}|\ge (1-2\cdot 2^{-n\ceil{\log q}/4})2^{n\ceil{\log q}}.
$

Our test of quantumness is described in Figure~\ref{fig:full}. In Subsection \ref{subsec:qprotocol} we explain how to pass the test when $q$ is large enough using a quantum prover that can be implemented in constant depth. In Subsection \ref{subsec:cprotocol} we then show that no classical computationally-bounded prover can pass this test with high probability under the $\LWE{}$ assumption, for a large range of parameters. A concrete test of quantumness can be obtained, for instance, by fixing $\varepsilon=1/n$, setting $B_L=\Theta(n)$, $m=\Theta(n^2)$, choosing $B_V$ superpolynomial in $n$ and taking $q=\Theta(B_Vn^{9/2})$. Theorem~\ref{th:quantum} shows that a small-depth quantum prover can pass the  corresponding test of quantumness with probability close to $1-1/n$, while Theorem~\ref{th:classical} shows that no polynomial-time classical prover can pass the test with probability significantly larger than~$3/4$, under the $\LWE{}$ assumption (the gap between the success probabilities of classical and quantum provers can easily be further amplified using parallel repetitions).  

\begin{figure}[ht!]
\begin{center}
\fbox{
\begin{minipage}{13.5 cm} 
Input: three positive integers $m$, $n$, $q$ such that $q\ge B_VC\sqrt{mn\log q}$ holds.
\begin{itemize}
\item[1.]
The verifier applies the procedure $\gentrap(1^n,1^m,q)$ and gets a pair $(A,t_A)$. The verifier then takes a vector $s\in\Int_q^n$ uniformly at random, and a vector $e\in \Int_q^m$ by sampling each coordinate independently according to the distribution $D_{q,B_V}$. The verifier sends the pair $(A,As+e)$ to the prover.
\item[2.]
The prover sends a vector $y\in\Int_q^m$ to the verifier.
\item[3.]
The verifier chooses a random bit $r$ uniformly at random and sends it to the prover.
\item[4.]
If $r=0$ then the prover sends a pair $(b,x)\in\{0,1\}\times\Int_q^n$ to the verifier. If $r=1$ then the prover sends a pair $(c,d)\in\{0,1\}\times \{0,1\}^{n\ceil{\log q}}$ to the verifier.
\item[5.]
If $r=0$ then the verifier accepts if and only if $\norm{Ax+bu-y}\le 2q/(C\sqrt{n\log q})$. 

\noindent 
If $r=1$, then the verifier applies the procedure $\invert(A,t_A,y)$ and get an output that we denote $x_0\in\Int_q^n$. The verifier accepts if and only if the three conditions $\norm{Ax_0-y}\le 2q/(C\sqrt{n\log q})$, $c=d\cdot (J(x_0)\oplus J(x_0-s))$ and $d\in G_{s,0,x_0}$ all hold.
\end{itemize}
\end{minipage}
}
\end{center}
\caption{Test of quantumness. Here $B_V>0$ is a parameter.}\label{fig:full}
\end{figure}


\subsection{Quantum protocol}\label{subsec:qprotocol}
Here is the main result of this subsection.
\begin{theorem}\label{th:quantum}
Let $\varepsilon$ and $B_V$ be any positive parameters.
There exists a quantum prover, which can be implemented by a circuit of class $\Cc_0(\Bb)$ acting on $\poly(m,n,\log q)$ qubits, that passes the test of Figure~\ref{fig:full} with probability at least $1-3\sqrt{\varepsilon}-\delta$ for all values $(m,n,q)$ such that $q\ge (8mB_VC\sqrt{mn\log q})/ \varepsilon$, where $\delta$ is some negligible function of the parameters.
\end{theorem}
\begin{proof}
The 5-tuple $(m,n,q,A,As+e)$ is in $\Kk_{B_V}$ with overwhelming probability (see the discussion after Theorem \ref{theorem:crypto} in Section \ref{sub:lattice}). We describe the quantum protocol under this assumption. After receiving the key at Step 1, the prover creates a state $\ket{\varphi}$ that is $q/(C\sqrt{mn\log q})$-bounded and $(\varepsilon/2,B_V)$-robust using Theorem~\ref{th:creation}. Then the prover applies Theorem \ref{th:preparation} using the state $\ket{\varphi}$ as input, which gives a state $\ket{\Phi}$ that is $\varepsilon$-close to some state in $\Hhh_k$.

Let us first describe and analyze the remaining of the protocol under the assumption that $\ket{\Phi}$ is in $\Hhh_k$ (instead of being only close to $\Hhh_k$). The prover measures the rightmost register of $\ket{\Phi}$. Let $y\in\Int_q^m$ denote the measurement outcome. The state after the measurement is 
\[
\ket{\Psi_y} = \frac{1}{\sqrt{2}}\left(\ket{0}\ket{x_0}+\ket{1}\ket{x_0-s}\right)\ket{y},
\]
where $x_0\in\Int_q^n$ is such that $\norm{Ax_0-y}\le q/(C\sqrt{n\log q})$. At Step 2, the prover sends this value~$y$.
At Step 4, if the prover received $r=0$, it measures the first two registers of the above state in the computational basis and simply sends to the verifier the measurement outcome $(b,x)$. This passes the verifier's check at Step 5 with certainty, since $\norm{Ax_0-y}\le q/(C\sqrt{n\log q})$ and $A(x_0-s)+u=Ax_0+e$, with
\[
\norm{Ax_0+e-y}\le q/(C\sqrt{n\log q})+\norm{e}\le q/(C\sqrt{n\log q})+B_V\sqrt{m}\le 2q/(C\sqrt{n\log q}).
\]

If the prover received $r=1$, it first applies an Hamadard gate on each qubit of the first two registers, which gives the state 
\[
\left(\frac{1}{2\sqrt{2^n}}\sum_{c\in\{0,1\}}\sum_{d\in\{0,1\}^n}\left((-1)^{J(x_0)\cdot d}+(-1)^{J(x_0-s)\cdot d+c}\right)\ket{c}\ket{d}\right)\ket{y}.
\]
The prover then measures the first two registers, and sends to the verifier the outcome $(c,d)$. Since $(c,d)$ necessary satisfies the equality $J(x_0)\cdot d \equiv J(x_0-s)\cdot d+c \:(\bmod\:  2)$, and $d\in G_{s,0,x_0}$ with overwhelming probability due to the density of $G_{s,0,x_0}$, the verifier's check succeeds at Step 5 with overwhelming probability, i.e., probability at least $1-\delta$ for some negligible function $\delta$.

Since the actual state $\ket{\Phi}$ is only $\varepsilon$-close to $\Hhh_k$ (instead of being in $\Hhh_k$ as we assumed so far), using the triangular inequality we can conclude that the success probability on the actual state is at least $1-\delta-\varepsilon-2\sqrt{\varepsilon}\ge1-\delta-3\sqrt{\varepsilon}$.
\end{proof}
\subsection{Classical hardness}\label{subsec:cprotocol}
In this subsection we will use exactly the same parameters and hardness assumption as in \cite{Brakerski+FOCS18}. 

Let $\lambda$ be a security parameter. All the other parameters are functions of $\lambda$. Let $q$ be a prime. Let $\ell,n,m\ge 1$ be polynomially bounded functions of $\lambda$, and $B_L$, $B_V$ be positive integers such that the following conditions hold:
\begin{itemize}
\item
$n=\Omega(\ell \log q)$ and $m=\Omega(n\log q)$,
\item
$2\sqrt{n}\le B_L<B_V\le q$,
\item
$B_V/B_L$ is superpolynomial in $\lambda$.
\end{itemize}

Here is the main result of this subsection.
\begin{theorem}\label{th:classical}
Assume a choice of parameters as above.
Assume the hardness assumption $\LWE{\ell,q,D_{q,B_L}}$ holds. No polynomial-time classical prover can pass the test of Figure \ref{fig:full} with probability greater than $3/4+\mu$, for some negligible function $\mu$ of the security parameter $\lambda$.
\end{theorem}
\begin{proof}
Consider a classical prover that passes the test with probability at least $3/4+\mu$ for some function $\mu$. 

Let us write $w$ the contents of the prover's memory and computation history at the end of Step 2 (note that $y$ can be recovered from $w$). Let $\Aa_0(w)$ be the algorithm the prover applies when it receives $0$ at Step 3, and $\Aa_1(w)$ be the algorithm the prover applies when it receives~$1$. Let consider the following strategy: Apply $\Aa_0(w)$ to get $(b,x)$, then rewind the computation and apply $\Aa_1(w)$ to get $(c,d)$, and finally output the 4-tuple $(b,x,d,c)$.

Let $p_0(w)$ denote the probability that the output of $\Aa_0(w)$ satisfies $\norm{Ax+bu-y}\le 2q/(C\sqrt{n\log q})$, and $p_1(w)$ denote the probability that the output of $\Aa_1(w)$ satisfies $c=d\cdot (J(x_0)\oplus J(x_0-s))$ and $d\in G_{s,0,x_0}$. Our assumption implies that $\E_w[p_0(w)/2+p_1(w)/2]\ge 3/4+\mu$. Thus the overall probability that $\norm{Ax+bu-y}\le 2q/(C\sqrt{n\log q})$, $c=d\cdot (J(x_0)\oplus J(x_0-s))$ and $d\in G_{s,0,x_0}$ all hold is at least 
\[
\E_w[1- (1-p_0(w))-(1-p_1(w))]=\E_w[(p_0(w)+p_1(w))-1]\ge  1/2+2\mu.
\]
In this case we have $x_0=x$ if $b=0$ and $x_0=x+s$ if $b=1$, and thus $c=d\cdot(J(x)\oplus J(x-(-1)^b s)$ holds in both cases.
Lemma 4.7 in \cite{Brakerski+FOCS18}, which we state for completeness in Appendix \ref{app}, guarantees that $\mu$ must be negligible.
\end{proof}

\bibliography{refs}
\appendix
\section{The Adaptive Hardcore Bit Lemma}\label{app}
For completeness, we reproduce below the statement of the adaptive hardcore bit lemma from \cite{Brakerski+FOCS18} on which the proof of Theorem \ref{th:classical} is based.
\begin{lemma}[Lemma 4.7 in \cite{Brakerski+FOCS18}]\label{th:LWE}
Assume a choice of parameters as in Section \ref{subsec:cprotocol}.
Assume the hardness assumption $\LWE{\ell,q,D_{q,B_L}}$ holds.
 Let $s\in\{0,1\}^n$. Write
\begin{align*}
H_s&=\left\{(b,x,d,d\cdot(J(x)\oplus J(x-(-1)^b s))\:|\: b\in\{0,1\},x\in\Int_q^n, d\in G_{s,b,x}\right\}\\
\overline{H}_s&=\left\{(b,x,d,c)\:|\:(b,x,d,c\oplus 1)\in H_s\right\}.
\end{align*}
Consider a pair $(A,As+e)$ generated as follows: generate $A$ using $\gentrap(1^n,1^m,q)$, then take $s\in\{0,1\}^n$ uniformly at random and~$e$ by sampling each coordinate independently according to the distribution $D_{\Int_q,B_V}$. Then for any polynomial-time algorithm $\Aa$ that receives as input the pair $(A,As+e)$ there exists a negligible function $\mu(\lambda)$ such that
\[
\Big|\Pr[\Aa(A,As+e)\in H_s] - \Pr[\Aa(A,As+e)\in \overline{H}_s]\Big|\le \mu(\lambda).
\]
\end{lemma}

\end{document}